\documentclass[a4paper, reqno]{amsart}

\usepackage[latin1]{inputenc}
\usepackage[T1]{fontenc}
\usepackage[english]{babel}
\usepackage{amsfonts}
\usepackage{amssymb}
\usepackage{amsmath}
\usepackage{amsthm}
\usepackage{graphicx,xcolor}
\usepackage{a4wide}

\newtheorem{proposition}{Proposition}
\newtheorem{lemma}{Lemma}
\newtheorem{remark}{Remark}
\newtheorem*{corollario}{Corollary}

\newtheorem*{Hyp}{H}
\newtheorem*{teorema}{Theorem}

\def\be{\begin{equation}}
\def\ee{\end{equation}}
\def\bea{\begin{eqnarray}}
\def\eea{\end{eqnarray}}
\def\ni{\noindent}
\def\nn{\nonumber}

\newcommand{\meanv}[1]{\left\langle#1\right\rangle}
\newcommand{\R}{\mathbb{R}}

\newcommand{\E}{\mathbb{E}}
\newcommand{\D}{\mathrm{d}}

\newcommand{\ie}{\textit{i.e. }}

\newcommand{\B}{\mathcal{B}}

\newcommand{\e}{\varepsilon}
\newcommand{\w}{\omega}

\def\eg{\textit{e.g.} }

\newcommand{\OOO}[1]{O \left(#1\right)}

\def\spect{\operatorname{Spect}}
\def\Var{\operatorname{Var}}

\def\s{\sigma}

\def\b{\beta}
\def\a{\alpha}
\def\g{\gamma}
\def\d{\delta}
\def\l{\lambda}
\def\r{\rho}

\title{Legendre Duality of Spherical and Gaussian Spin Glasses}
\date{\today}

\author{Giuseppe Genovese}
\address{Giuseppe Genovese: Institut f\"ur Mathematik, Universit\"at Z\"urich,
Winterthurerstrasse 190, CH-8057 Z\"urich, Switzerland.}
\email{giuseppe.genovese@math.uzh.ch}

\author{Daniele Tantari}
\address{Daniele Tantari: Dipartimento di Matematica, Sapienza Universit\`a di Roma, Piazzale Aldo Moro 2, 00185, Roma, Italia.}
\email{tantari@mat.uniroma1.it}

\subjclass[2000]{%
}\keywords{}

\begin{document}\maketitle

\begin{abstract}
\ni The classical result of concentration of the Gaussian measure on the sphere in the limit of large dimension induces a natural duality between Gaussian and spherical models of spin glass. We analyse the Legendre variational structure linking the free energies of these two systems, in the spirit of the equivalence of ensembles of statistical mechanics. Our analysis, combined with the previous work \cite{BGGT}, shows that such models are replica symmetric. Lastly, we briefly discuss an application of our result to the study of the Gaussian Hopfield model.

\vspace{5mm}

\ni \textbf{MSC:} 82B44.
\end{abstract}
\section{Introduction and Main Result}

\ni The equivalence between Gaussian measure and uniform measure on the sphere in the limit of large dimension, is nowadays a classical argument, shared by probability and mathematical physics. It goes back traditionally to Poincar\'e and we refer to the paper \cite{persi} for a detailed mathematical and historical discussion. 

\ni Roughly, the probabilistic idea is that the spherical measure of cylindrical sets approaches the Gaussian one when the dimension becomes infinite (see \eg \cite{Mk}). Physically this means that for a gas of non interacting particles, therefore with a fixed kinetic energy, the single particle velocity is distributed according to the Maxwell-Boltzmann statistics in the thermodynamic limit. 

\ni For spin systems, the intimate connection between Gaussian and spherical models has been noticed since their first systematic introduction by Berlin and Kac in \cite{BK} for ferromagnetic systems. In the present paper we are going to investigate this relation for spin glasses. 

\ni We will consider a system of $N$ soft spin $z_i\in\R$, $i=1...N$ interacting through the mean field disordered Hamiltonian:
\be\label{eq:H}
H_N(z,J):=-\frac{1}{\sqrt{N}}\sum_{(i,j)}J_{ij} z_iz_j.
\ee
In the whole paper we will use the following hypothesis on the disorder:

\begin{Hyp}\label{H}
The random matrix $J_{ij}$ is in the Symmetric Wigner Ensemble. In addition we assume that there is a $\theta>0$ such that $\mathbb{E}\left[e^{\theta \left(\frac{J_{ij}}{\sqrt{N}}\right)^2}\right]<\infty$ $\forall i,j$ $N$-uniformly. 
\end{Hyp}

\ni Then (see for instance \cite{AGZ}\cite{Tao}) $\left\{J_{ij}/\sqrt{N}\right\}$ can be diagonalised and $\spect[J]:=\spect\left[\{\frac{J_{ij}}{\sqrt{N}}\}\right]\in\R$. Furthermore $\mathbb{E}[J_{ij}]=0$ and $\mathbb{E}[J_{ij}J_{hk}]=J^2\d_{ih}\d_{jk}$ for a certain constant $J^2>0$, and
\begin{enumerate}
\item There is a $\bar\l>0$ such that $\forall a>\bar\l$ 
\be\label{eq:P(l>a)}
P(|\l|\geq a)\leq C_1 e^{-\theta a^2 N},
\ee
for any $\l\in\spect[J]$ and two constants $C_1,\theta>0$;\\
\item The distribution of eigenvalues of $\left\{J_{ij}/\sqrt{N}\right\}$ converges for $N\to\infty$ to the semicircle law $$\r(\l)=\frac{2\sqrt{\bar\l^2-\l^2}}{\pi \bar\l^2}.$$
\end{enumerate}
We will name the constant $\bar\l:=\max\spect[J]=\sqrt{J^2}$ the maximum eigenvalue of $\left\{J_{ij}/\sqrt{N}\right\}$.

\ni We will be concerned here about two kind of distributions for the $N$ continuous spin $z_1, \dots, z_N$, namely the uniform distribution on the $N$ dimensional sphere in $\R^N$, $S_{R\sqrt{N}}$, centred in the origin with radius $R\sqrt{N}$, or spherical distribution $\s_{R,N}(z)$ and the standard Gaussian distribution $\g_N(z)$. 

\ni Our interest in the topic comes from the theory of neural networks. In collaboration with A. Barra and F. Guerra we have proven that, in the case of Gaussian distributed patterns, the free energy of the Hopfield model can be written as a convex sum of the free energy of the Sherrington-Kirkpatrick (SK) model and a suitably defined Gaussian spin glass \cite{NN?}. This decomposition holds exactly in the high temperature regime and in the replica symmetric approximation. We stress that the same feature had been already observed in other bipartite spin glasses (see \cite{Bip}). 

\ni For these reasons in \cite{BGGT} a model of Gaussian distributed spins with disordered interaction has been studied . More precisely, let $z_i\in\R$, $i=1...N$, be i.i.d. random (soft) spin  $\mathcal{N}(0,1)$ variables, and let them interact through the Hamiltonian (\ref{eq:H}). In such a model the divergence of the partition function arises quite naturally at low temperature, so one needs to introduce the regularised 
\be\label{eq:Zg}
 Z^{g}_N(\beta,\lambda):=\int_{\R^N}dz_1...dz_N\frac{e^{-\|z\|^2/2}}{(2\pi)^{N/2}}e^{\left(-\beta H_N(z,J)-\frac{\beta^2}{4N}\|z\|^4+\frac{\lambda}{2}\|z\|^2\right)},
\ee
for $\l\in\R$, such that $\E Z(\b,J,0)=1$ for Gaussian distributed disorder. The associated quenched pressure is
\be\label{eq:Agauss-def}
A_N^g(\b,\l):=\frac1N\mathbb{E}\log  Z^{g}_N(\beta,J,\lambda),\qquad A^g(\b,\l):=\lim_N\frac1N\mathbb{E}\log  Z^{g}_N(\beta,J,\lambda).
\ee 

\ni In \cite{BGGT} we have approached the problem from a nowadays usual perspective in spin glass theory, after the celebrated results by Guerra and Talagrand on the SK model (for which we refer to \cite{leshouches} and \cite{talabook}). We have studied the Edward Anderson order parameter, \ie the replicas overlap, by using Guerra's interpolation. The main achievement contained in the paper is that the broken replica symmetry (RSB) bound does not improve the replica symmetric (RS) one, that is:
\be\label{eq:AgRS}
A^{g}_{RS}(\b,\l)=-\frac{1}{2}\log((1-\lambda+\beta^2\bar{q}))+\frac{\beta^2\bar{q}}{(1-\lambda+\beta^2\bar{q})}+\frac{\beta^2}{4}\bar{q}^2,
\ee
with the RS order parameter given by $\bar{q}=0$ for $\beta\leq 1-\lambda$ and $\bar{q}=\frac{\beta-(1-\l)}{\beta^2}$ otherwise.

\ni Originally, at least on the mathematical side, this model has been introduced and completely solved by Ben Arous, Dembo and Guionnet in \cite{BDG}. In the first part of their paper the equilibrium properties of the model are analysed via the construction of a large deviation principle, which is used to study the aging phenomenon in the Langevin dynamics in the second part. Incidentally, as a byproduct of our analysis, we provide here a different strategy to get the formula for the free energy obtained in \cite{BDG}. 

\ni On the other hand, the spherical model of spin glass was introduced by Kosterlitz, Thouless and Jones in \cite{KT}, where the authors gave the form of the free energy. It turns out to be RS and their method, although the proof passes over some mathematical details, is rigorous. Then Crisanti and Sommers have studied the $p$-spin case in \cite{CS} and successively Talagrand has proved in all the details the validity of the general Crisanti-Sommers solution in \cite{Tsfer}.

\ni The Spherical Model is defined as follows: let $z_i\in\R$, $i=1,\dots,N$ be $N$ i.i.d. random spin variables distributed according to $\s_{R,N}$, and let them interact via the Hamiltonian (\ref{eq:H}). The partition function is defined as
\be\label{eq:Zsf}
Z^{sf}_N(\beta,R):=\int_{\R^N}d\s_{R,N}(z)e^{-\beta H_N(z,J)}.
\ee
To this partition function it is associated the pressure:
\be\label{eq:Asf}
A^{sf}_N(\b,R):=\frac{1}{N}\E\log Z^{sf}_{N},\qquad A^{sf}(\b,R):=\lim_N \frac{1}{N}\E\log Z^{sf}_{N}.
\ee
Since it depends in fact by $R^2$, we will also denote the pressure as $A^{sf}(\b,R^2)$ (in the literature one finds $\bar\l, R=1$). The results of Crisanti, Sommers and Talagrand, along with later works on the subject (\cite{PT}\cite{Panc}\cite{franz}), deal with Gaussian disorder. However we will see that this assumption can be relaxed in the sense of hypothesis \textbf{H} (see also \cite{BDG}). Observe that these models seem to be more sensitive to the disorder distribution with respect to, for instance, the SK model.

\vspace{0.5cm}

\ni The main result of the present work is the following

\begin{teorema}
The pressure of the spherical model converges in the thermodynamic limit a.s. to 
\be\label{eq:A-sferico}
A^{sf}(\b,R^2)=\min_{q\geq\b\bar\l}\left(q R^2-\frac 12\int\r(\l)\log(q-\b\l)-\frac12\log R^2-\frac12-\frac12\log2\right).
\ee
The pressure of the Gaussian model converges in the thermodynamic limit a.s. to
\be\label{eq:A-gauss}
A^g(\b,\l)=\max_{R^2\in(0,\infty)}\left(A^{sf}(\b,R^2)-\frac{\b^2R^4}{4}+\frac{(\l-1)R^2}{2}+\frac12\log R^2+\frac{1}{2} \right).
\ee
\end{teorema}

\vspace{0.5cm}

\begin{remark}
Formula (\ref{eq:A-gauss}) can be interpreted as an ordinary Legendre transformation between $A^g(\b, \l)$ and $A^{sf}(\b, R^2)$. As it can be easily verified, the Legendre transformation is well-defined and involutive, and we have the inverse formula
\be\label{eq:dual-rel}
A^{sf}(\b,R^2)=\min_{\l\in\R}\left(A^g(\b,\l) +\frac{\b^2R^4}{4}+\frac{(1-\l)R^2}{2}-\log (R)-\frac{1}{2} \right).
\ee
In this way the duality between the two models is completely specified.
\end{remark}

\begin{remark}
By direct calculations from (\ref{eq:A-sferico}) we obtain the following explicit expressions for $A^{sf}(\b, R)$:
\be\label{eq:A_sf-expl}
A^{sf}(\b, R^2)=\left\{
\begin{array}{ccc}
\frac14\left(\frac{\b}{\b_c}\right)^2&\quad&\b\leq\b_c,\\
\frac{\b}{\b_c}-\frac12\log\left(\frac{\b}{\b_c}\right)-\frac34&\qquad&\b\geq\b_c,
\end{array}\right.
\ee
with a discontinuity of the third derivative in $\b_c:=(\bar\l R^2)^{-1}$; then, by plugging (\ref{eq:A_sf-expl}) in (\ref{eq:A-gauss}), after the optimisation we get
\be \label{gaussian free energy}
A^g(\b, \l)=\begin{cases}
 -\frac {1}{2} \log(1-\l)   & \text{$\b<1-\l$},\\
 -\frac{1}{2}\log(\b)+\frac{\b\bar{q}}{2}+\frac{\b^2\bar{q}^2}{4},   & \text{$\b\geq1-\l$},
\end{cases}
\ee 
with $\bar{q}(\b,\l):=\frac{\b-(1-\l)}{\b^2}$. 
\end{remark}

\ni Therefore, by a direct comparison with the results in \cite{BGGT}, we can conclude that both the spherical and the Gaussian models are entirely replica symmetric.

\vspace{0.5cm}

\ni In Section \ref{sec:proof} we will prove the theorem. Our strategy consists of three steps:

\vspace{0.2cm}

\begin{enumerate}
\item We obtain the variational formula (\ref{eq:A-sferico}) for the pressure of the spherical model;
\item We introduce suitable extensions and restrictions respectively of the Spherical and Gaussian model to a spherical shell and show the uniform convergence of the pressure of the spherical shell model to the spherical one;
\item We prove the Legendre duality (\ref{eq:A-gauss}), by slicing $\R^N$ into spherical shells and proving concentration of the Gaussian Gibbs measure on a particular one. This allows to relate the pressures of the Gaussian and spherical models on the shells and then to get (\ref{eq:A-gauss}), taking properly the limits.
\end{enumerate}

\vspace{0.2cm}

\ni Finally, in Section \ref{sec:concl}, we will discuss some implications of our result on the representation of the pressure for the Hopfield model obtained in \cite{NN?} and we will add some conclusive remarks. 

\vspace{0.5cm}

\ni With a little abuse of notation, throughout the paper we will indicate the random Gibbs state always with $\w$, without distinguish the Gaussian or spherical Gibbs measure. However it will be always clear by the context to which one this symbol is referred. Furthermore the area of the spherical surface of radius $R$ in $\R^N$ $S_R$ will be denoted by $|S_R|$.

\vspace{0.5cm}

\section{Proof of the Theorem}\label{sec:proof}

\vspace{0.3cm}
\subsection{Pressure of the Spherical Model}

\ni Our first goal is to obtain formula (\ref{eq:A-sferico}), that appeared originally in \cite{KT}. 

\ni In primis we diagonalise the interaction, as it is usual for this kind of models \cite{KT}\cite{BDG}\cite{FS}, in virtue of their rotational symmetry:
\be
H_N\longrightarrow-\sum_i \l_i z_i^2.
\ee

\ni Let us introduce the annealed pressure
\be\label{eq:Asf-ann-DEF}
a^{sf}(\b,R):=\lim_N\frac1N\log\E Z_N^{sf}.
\ee

\ni To begin with, we get a rough bound on the annealed pressure, and so (by Jensen inequality) to the quenched one.
It will result helpful to define the events
$$
\mathcal{B}_a:=\left\{\max_{i=1,\dots,N}|\l_i|< a\right\}
$$
for fixed $N$ (that we omit in the notations) and each $a\geq\bar\l$, and its complementary $\mathcal{B}^c_a$. In the following we will denote with $\mathcal{I}_B$ the indicator function of the set $B$. Thus we have the following
\begin{proposition}\label{lem-annealed-sf}
The pressure of the spherical model is bounded and
\be\label{eq:Ann-Bound-SF}
\limsup_NA_N^{sf}\leq a^{sf}\leq\max\left(\b\bar\l R^2,\frac{\b^2R^4}{4\theta}\right).
\ee
\end{proposition}
\begin{proof}
Let us fix any $a>\bar\l$. It is
$$
\E[e^{\b\sum_i\l_i z_i^2}]=\E[e^{\b\sum_i\l_i z_i^2}\mathcal{I}_{\B_a}]+\E[e^{\b\sum_i\l_i z_i^2}\mathcal{I}_{\B^c_a}].
$$
For the first addendum on the r.h.s. we easily get the bound
\be
\E[e^{\b\sum_i\l_i z_i^2}\mathcal{I}_{\B_a}]\leq e^{\b R^2 aN}.
\ee
For the second one we decompose $\B^c_a=\bigcup_{k\geq1} B_k$ with
$$
B_k:=\left\{ak\leq\max_{i=1,\dots,N}|\l_i|\leq a(k+1)\right\}.
$$
Therefore using (\ref{eq:P(l>a)}) we have
\bea
\E[e^{\b\sum_i\l_i z_i^2}\mathcal{I}_{\B^c_a}]&=&\sum_{k\geq1} \E[e^{\b\sum_i\l_i z_i^2}\mathcal{I}_{B_k}]\nn\\
&\leq&\sum_{k\geq1} e^{\b R^2Nak} \E[\mathcal{I}_{B_k}]\nn\\
&\leq&C_1\sum_{k\geq1} e^{\b R^2Nak-\theta N a^2k^2}\nn\\
&\leq&C_1a\sqrt{2\pi\theta N}e^{\b^2 R^4N/4\theta}\nn.
\eea

\ni Hence, neglecting the terms vanishing in the limit, we readily get for every $a>\bar\l$
\bea
a^{sf}(\b,R)&\leq&\lim_N\frac1N\log\left(e^{\b R^2 aN}+C_1a\sqrt{2\pi\theta N}e^{\b^2 R^4N/4\theta}\right) \nn\\
&=&\max\left(\b R^2 a,\frac{\b^2R^4}{4\theta}\right)\nn,
\eea
The inequality is satisfied also by taking the infimum on $a>\bar\l$, whence (\ref{eq:Ann-Bound-SF}) follows.
\end{proof}

\ni The first consequence of the previous Proposition is the following
\begin{lemma}\label{lemma:dis-conc}
We have
\be\label{eq:lemma2}
\mathbb{E}[A_N^{sf}(\b,R)\mathcal{I}_{\mathcal{B}^c_{\bar\l}}]=\OOO{e^{-N}}.
\ee
\end{lemma}
\begin{proof}
At first we note that the pressure (before expectation) has bounded gradient w.r.t. the $\l_i$ with full probability, since
\be\label{eq:bound-grad-A-SF}
\frac1N\partial_{\l_i} \log Z_N^{sf}=\frac\b N\omega(z^2_i).
\ee
Then we have
\bea
(\E[A^{sf}_N\mathcal{I}_{\B^c_{\bar\l}}])^2&\leq& P(\B^c_{\bar\l}) \E[A_N^2]\nn\\
&=&P(\B^c_{\bar\l})(\Var_{\r_N}(A^{sf}_N)+(\E[A_N^{sf}])^2)\nn\\
&\leq&P(\B^c_{\bar\l})(\|\nabla A^{sf}_N\|_2^2+(\E[A_N^{sf}])^2))\nn\\
&\leq&P(\B^c_{\bar\l})(\b^2 R^4+a^2_{sf})\nn,
\eea
where we have exploited Poincar\'e inequality w.r.t. the empirical measure of the eigenvalues and (\ref{eq:bound-grad-A-SF}). Due to Proposition $\ref{lem-annealed-sf}$, we obtain (\ref{eq:lemma2}) by the decay (\ref{eq:P(l>a)}).
\end{proof}

\begin{remark}
Because of Lemma \ref{lemma:dis-conc} the partition function can be easily bounded by
$$
Z_N^{sf}(\b,R)\leq e^{\b\bar\l R^2 N}.
$$
Therefore
$$
\frac1N\mathbb{E}\left[\log Z^{sf}_N(\b,R)\Big|\max_{i=1,...,N}|\l_i|\leq\bar \l\right]\leq\bar\l R^2,
$$
and so
\be\label{eq:boundSF}
\limsup_NA^{sf}_N(\b,R)\leq\b\bar\l R^2.
\ee
\ni This bound, albeit still rather course, refines the annealed one for small temperature.
\end{remark}

\ni However we can do much better than that, proving directly formula (\ref{eq:A-sferico}). In the past literature, the common way to face this kind of problems relied on a direct calculation of the partition function. Several techniques can be implemented for this task: the original one by Berlin-Kac (see \cite{BK}, appendix B) and a variant by Montroll \cite{Mon} make use essentially of Riemann steepest descendent method; alternatively, the moment expansion method developed by Von Neumann in \cite{VN} certainly deserves to be mentioned. Here we present a different and purely variational proof, which captures in our opinion the two essential aspects of the model: 1) only the largest eigenvalue determines the form of the free energy; 2) thermodynamics naturally forces the equilibrium configurations of the system on the sphere, even if we relax the spherical constraint. 

\begin{proof}[Proof of (\ref{eq:A-sferico})]

\ni From Lemma \ref{lemma:dis-conc}, we can limit ourself to consider only realisation of the disorder with spectrum contained in an interval $[-\bar\l,\bar\l]$ with full probability. Then for every $q>\b\bar\l$ we have
\bea
Z_{N}^{sf}(\b,R)&=&e^{q R^2 N}\int_{\R^N}d\s_{R,N}(z)e^{-\sum_i^N(q-\b\l_i)z_i^2}\nn\\
&\leq&e^{q R^2 N}\frac{(2\pi)^{\frac N2}}{|S_{R\sqrt{N}}|}\int_{\R^N}\frac{d^Nz}{(2\pi)^{\frac N2}} e^{-\sum_i^N(q-\b\l_i)z_i^2}\nn\\
&=&e^{q R^2 N}\frac{(2\pi)^{\frac N2}}{|S_{R\sqrt{N}}|}e^{-\frac 12\sum_i^N\log2(q-\b\l_i)}\nn
\eea
and so
\be
\limsup_NA^{sf}_N(\b)\leq q R^2 -\frac 12\int\r(\l)\log(q-\b\l)-\log R-\frac12-\frac12\log2=:\tilde A(q).
\ee

\ni We note that for $q>\b\bar\l$
$$
\partial_q^2 \tilde A(q)=\frac12\int d\l\frac{\r(\l)}{(q-\b\l)^2}>0,
$$
and so the functional $\tilde A(q)$ is uniformly convex (independently on $R$). Furthermore we can explicitly verify that $\tilde A(q)$ is continuous for $q\to\b\bar\l$, thus
$$
\limsup_NA^{sf}_N(\b)\leq\min_{q\geq\b\bar\l}\tilde A(q).
$$
Since the functional on the r.h.s. is uniformly convex in $q$, there is a unique point $\bar q$ where the minimum is attained. 

\ni The reverse bound is slightly less direct. Let us consider for each $\e>0$ the spherical shell around the radius $R\sqrt{N}$ (as defined in (\ref{eq:shell})) and its complementary set, denoted by ${S^\e}^c$. Since $S_{R\sqrt{N}}^\e\cup {S^\e}^c=\R^N$, it holds
$$
\e Z^{sh}_{\e N}=e^{\bar q R^2 N}\frac{(2\pi)^{\frac N2}}{|S_{R\sqrt{N}}|}\int_{\R^N}\frac{d^Nz}{(2\pi)^{\frac N2}} e^{-\sum_i^N(q-\b\l_i)z_i^2}-e^{\bar q R^2 N}\frac{(2\pi)^{\frac N2}}{|S_{R\sqrt{N}}|}\int_{{S^\e}^c}\frac{d^Nz}{(2\pi)^{\frac N2}}z e^{-\sum_i^N(q-\b\l_i)z_i^2}.
$$
We can use the Chernoff bound to estimate the second addendum. In fact, for every $\mu,\eta>0$ we have
\bea
\int_{{S^\e}^c}\frac{d^Nz}{(2\pi)^{\frac N2}}z e^{-\sum_i^N(q-\b\l_i)z_i^2}&\leq&\exp\left[N\left(\mu\left(R^2-\frac\e N\right)-\frac{1}{2N}\sum_j\log(q-\b\l_j+\mu)-\frac12\log2\right)\right]\nn\\
&+&\exp\left[N\left(-\eta\left(R^2+\frac\e N\right)-\frac{1}{2N}\sum_j\log(q-\b\l_j-\eta)-\frac12\log2\right)\right]\nn.
\eea
\ni The r.h.s. of the last inequality is $o(e^{-N})$ if $\frac\e N\to\infty$ when $N\to\infty$, so it does not contribute to the thermodynamics. Therefore we can neglect $\e$ growing at those scales, by setting $\e=\tilde\e N$, $\tilde\e$ positive and independent by $N$. Thus, by a straightforward computation, we get as $N\to\infty$ 
\be\label{eq:max( , , )}
\liminf_{N}A^{sh}_{\tilde\e,N}\geq\max\left(\tilde A(q), A^1_{\tilde\e}(\mu; q), A^2_{\tilde\e}(\eta; q)\right),
\ee
with
\bea
A^1_{\tilde\e}(\mu; q)&=&(q+\mu)R^2-\tilde\e\mu-\frac12\int d\l\r(\l)\log(q-\b\l+\mu)-\frac12-\frac12\log2-\log R;\\
A^2_{\tilde\e}(\eta; q)&=&(q-\eta)R^2-\tilde\e\eta-\frac12\int d\l\r(\l)\log(q-\b\l-\eta)-\frac12-\frac12\log2-\log R.
\eea
Our aim is to show that, for some $q$, $\tilde A(q)$ is greater than these other two quantities. We have
\bea
\D_1(q;\mu)&:=&\tilde A(q)-A^1_{\tilde\e}(\mu; q)= -\mu(R^2-\tilde\e)+\frac12\int d\l\r(\l)\log\left(\frac{q-\b\l+\mu}{q-\b\l}\right);\\
\D_2(q;\mu)&:=&\tilde A(q)-A^2_{\tilde\e}(\eta; q)= \eta(R^2+\tilde\e)+\frac12\int d\l\r(\l)\log\left(\frac{q-\b\l-\eta}{q-\b\l}\right).
\eea
\ni Let us regard for instance to $\D_1(q;\mu)$ as a function of $\mu$: it is continuos and derivable, it vanishes in $\mu=0$ and it goes to $-\infty$ for $\mu\to+\infty$. So it can assume positive values (in particular a positive maximum) if and only if the derivative in $\mu=0$ is positive, that is
\be\label{eq:upper-bound-epsi1}
0<-(R^2-\tilde\e)+\frac12\int d\l\frac{\r(\l)}{q-\b\l}=\tilde \e-\partial_q \tilde A(q),
\ee
where we have used that in the $\mu-$ derivative of $\D_1(q;\mu)$ it appears exactly the derivative $\partial_q \tilde A(q)$ (see the explicit analysis below).

\ni Analogously  $\D_2(q;\eta)$ is zero in the origin and it approaches $+\infty$ for $\eta\to+\infty$. Thus it is always positive, provided that $\left.\partial_\eta \D_2(q;\eta)\right|_{\eta=0}\geq0$, \ie
\be\label{eq:upper-bound-epsi2}
0\leq(R^2+\tilde\e)+\frac12\int d\l\frac{\r(\l)}{q-\b\l}=\tilde \e+\partial_q \tilde A(q).
\ee
Conditions (\ref{eq:upper-bound-epsi1}) and (\ref{eq:upper-bound-epsi2}) must be satisfied together; so we seek a $\tilde q$ such that for any $\tilde\e>0$ it is
$$
-\tilde\e\leq \left.\partial_q \tilde A(q)\right|_{q=\tilde q}<\tilde\e.
$$
This simply means that $\tilde q=\bar q$, viz. the unique stationary point of $\tilde A(q)$. With this choice of $q$, relation (\ref{eq:max( , , )}) gives
\be
\liminf_N A^{sh}_{\tilde \e,N}(\b)\geq \min_{q\geq\b\bar\l}\tilde A(q)\quad{\tilde\e\mbox{ uniformly}},
\ee
so we can send $\tilde\e\to0$ obtaining (\ref{eq:A-sferico}).

\ni Let us look more thoroughly to $\tilde A(q)$. The first derivative reads
$$
\partial_q \tilde A(q)=R^2-\frac12\int d\l\frac{\r(\l)}{q-\b\l};
$$
hence we see that for $\b<(\bar\l R^2)^{-1}$ the derivative changes sign from negative to positive in a point $\bar q$ given by the equation
\be\label{eq:q-Sferico}
\sqrt{\left(\frac{\bar qR^2}{2}\right)^2-\left(\frac{\b\bar\l R^2}{2}\right)^2}=\frac{\bar qR^2}{2}-\frac{\b^2\bar\l^2 R^4}{2},
\ee 
which is solved by $\bar q=\frac{1}{2R^2}\left(1+\b^2\bar\l^2 R^4\right)$.

\ni On the other hand for $\b>\frac{1}{\bar\l R^2}$ the derivative is always a positive function (equation (\ref{eq:q-Sferico}) is never satisfied). This means that the minimum is attained in the extremum of the interval of definition, \ie $\bar q=\b\bar\l$. Thus we have that the critical point is defined by $\b_c=\frac{1}{\bar\l R^2}$ as a singular point of the minimiser function of $\tilde A(q)$. 

\ni Finally we notice that, since $A_N$ is a convex and Lipschitz continuous function of $\l_i$, with a certain constant $L(\bar\l)$, it has to satisfy Talagrand inequality
\be\label{eq:A-sf-a.s.Tal}
P\left(\left|\frac1N\log Z_N^{sf}-A_N\right|\geq\e\right)\leq e^{-N\frac{\e^2}{L(\bar\l)^2}}.
\ee
So we get convergence in probability and by Borel--Cantelli lemma also convergence a.s. follows.
\end{proof}


\vspace{0.3cm}

\subsection{Models on Spherical Shells}

\ni It results useful to introduce the spherical and Gaussian models on a spherical shell. We consider the spherical shell of radius $R>0$ and thickness $\e>0$
\be\label{eq:shell}
S^{\e}_{R}:=\left\{z_1,...,z_N\in R^N: R-\frac{\e}{2} <\|z\|\leq R+\frac{\e}{2}\right\},
\ee
such that $\bigcup_{R}S^{\e}_R=\R^N$ $\forall \e>0$. The spherical shell partition function is defined as
\be\label{eq:Zsfershell}
Z^{sh}_{N,\e}(\b,J,R_N):=\frac{1}{\e}\int_{S^{\e}_{R_N}}\frac{dz_1...dz_N}{|S_{R_N}|}e^{-\beta H_N(z,J)},
\ee
where $H_N(z,J)$ is given by (\ref{eq:H}) and $R_N$ is a given sequence of radii. It turns out to be a fuzzy version of the spherical model, since, for $R_N=R\sqrt{N}$, 
\be\label{eq:lim-e-Zsh-Zsf}
\lim_{\e\to0} Z^{sh}_{N,\e}(\b,J,R\sqrt{N})=Z^{sf}_{N}(\b,J,R).
\ee
We can alternatively define
$$
Z^{sh}_{N,\e}(\beta,J,R):=\int_{\R^N}dz_1...dz_N\frac{\d^{\e}(\|z\|=R\sqrt{N})}{|S_{R\sqrt{N}}|}e^{-\beta H_N(z,J)},
$$
denoting as $\d^{\e}$ the generic mollified projector on the sphere. The two approaches are equivalent and they become identical if we take 
$$
\d^{\e}=\chi([R-\e/2,R+\e/2])/\e,
$$ 
where $\chi(\cdot)$ is the characteristic function of an interval in the radial coordinates. We will use this notation below.

\ni The pressure of the model is given by
\be\label{eq:Asfshell}
A^{sh}_{N,\e}(\b,R_N)=\frac{1}{N}\E\log Z^{sh}_{N,\e},
\ee
and we have straightforwardly 
\be\label{eq:Asfshell}
\lim_{\e\to0}A^{sh}_{N,\e}(\b,R\sqrt{N})=A^{sf}_N(\b,R)\,\Rightarrow\, \lim_N\lim_{\e\to0}A^{sh}_{N,\e}(\b,R\sqrt{N})=A^{sf}(\b,R),
\ee
given by (\ref{eq:A-sferico}). The next Lemma allows us to exchange to limits.

\begin{lemma}\label{lem:chlim} 
The $N\to \infty$ and $\e\to 0$ limits for the pressure of the spherical shell model can be taken in any order, and we have
\be\label{limitchange}
\lim_{N\to \infty}\lim_{\e\to0}A^{sh}_{N,\e}(\b,R\sqrt{N})=\lim_{\e\to0}\lim_{N\to \infty}A^{sh}_{N,\e}(\b,R\sqrt{N})=A^{sf}(\b,R)
\ee 
\end{lemma}

\begin{proof}

\ni By the mean-value theorem of integration, we can write
\be
A^{sh}_{N,\e}(\b,R_N)=A^{sf}_N(\b,R_{N,\e}),
\ee
for some $R_{N,\e}\in [R_N-\e/2,R_N+\e/2]$. In this way, setting $R_N=R\sqrt{N}$, we can estimate the difference
\bea
\left|A^{sh}_{N,\e}(\b,R_N)-A^{sf}_N(\b,R)\right|&=&\left|A^{sf}_{N}(\b,R_{N,\e})-A^{sf}_N(\b,R)\right|\nn\\
&=&\left|\frac 1 N \E \log \left(\frac{Z^{sf}_N(\b,J,R_{N,\e})}{Z^{sf}_N(\b,J,R)}\right)\right|.
\eea
Now, by using the properties of the spherical integral, we can turn the problem of integration on a different radius into a more treatable shift in temperature. We have
\be 
Z^{sf}_N(\b,J,R_{N,\e})=\int_{\|z\|=R_{N,\e}}\frac{dz}{|S_{R_{N,\e}}|}e^{-\b H_N(z,J)}=
\int_{\|z\|=R_{N}}\frac{dz}{|S_{R_{N}}|}e^{-\b\frac{R^2_{N,\e}}{R^2_N} H_N(z,J)},\nn
\ee
so that
\be 
\frac 1 N \E \log \left(\frac{Z^{sf}_N(\b,J,R_{N,\e})}{Z^{sf}_N(\b,J,R)}\right)=
\frac 1 N \E \log \omega\left( e^{-\b C^{\e}_N H}\right),
\ee
where we have isolated the $\e$-dependence in the term $C^{\e}_N:=\frac{R^2_{N,\e}}{R^2_N}-1=O(\e/\sqrt{N})$. Finally, because of Lemma \ref{lemma:dis-conc} we have
\be
\frac 1 N \E \log \omega \left(e^{-\b C^{\e}_N H}\right)=\frac 1 N \E \log \omega\left( e^{\b C^{\e}_N \sum_{i=1}^N\l_i z_i^2}\right)\leq\b R^2 \bar{\l} C^\e_N.
\ee
Therefore
\bea 
\left|A^{sh}_{N,\e}(\b,R_N)-A^{sf}_N(\b,R)\right|&=&\left|A^{sf}_{N}(\b,R_{N,\e})-A^{sf}_N(\b,R)\right|\nn\\
&=& O(\e/\sqrt{N}),
\eea
whence the Lemma follows.
\end{proof}

\begin{remark}
We notice that the second equality of ($\ref{limitchange}$) holds even before taking the $\e\to 0$ limit. 
Consequentially, we have also shown that, as $N\to\infty$, the spherical shell model (of radius $R\sqrt{N}$ and thickness $\e$) and any other spherical model defined inside the shell have the same free energy of a spherical model with radius $R\sqrt{N}$.
\end{remark}

\ni In analogy we can restrict the Gaussian model on the shell with  a proper cut off:
\bea
Z^{gsh}_{N,\e}(\b,J,\l,R_N)&:=&\frac{1}{\e}\int_{S^{\e}_{R_N}}\frac{dz_1...dz_N}{(2\pi)^{N/2}}e^{\left(-\beta H_N(z,J)-\frac{\beta^2}{4N}\|z\|^4+\frac{(\lambda-1)}{2}\|z\|^2\right)},\nn\\
A^{gsh}_{N,\e}(\b,\l,R_N)&:=&\frac{1}{N}\E\log Z^{gsh}_{N,\e}.\label{eq:A-gSh}
\eea
As in the previous case we have that
\be
Z^{gsh}_{N,\e}(\b,J,\l,R_N)=\int_{\R^N}\frac{dz_1...dz_N}{(2\pi)^{N/2}}\d^{\e}(\|z\|=R_N) e^{\left(-\beta H_N(z,J)-\frac{\beta^2}{4N}\|z\|^4+\frac{(\lambda-1)}{2}\|z\|^2\right)}\label{eq:Zgsh}.
\ee
Let us consider now, for a certain $R>0$, a spherical and a Gaussian model on the same spherical shell centered on $R_N=R\sqrt{N}$. Since we can fix the term $\|z\|^2=R^2N$ in the Gibbs measure of the Gaussian model, with a small error $O(\e)$,
\be
A_{N,\e}^{gsh}(\b,\l,R\sqrt{N})-A_{N,\e}^{sh}(\b,R\sqrt{N})=-\frac{\b^2}{4}R^4+\frac{\l-1}{2}R^2+\frac{1}{N}\log\left(\frac{S_N}{\sqrt{2\pi}^N}\right) +O(\e /N) \nn
\ee
and, in virtue of Lemma \ref{lem:chlim}, 
\be\label{eq:rel-Asf-Ag}
A^{gsh}_{N,\e}(\b,\l,R\sqrt{N})=A_N^{sf}(\b,R)-\frac{\b^2}{4}R^4+\frac{\l-1}{2}R^2+\log R+\frac{1}{2} + o_N(1),
\ee
where $o_N(1)$ stands for a vanishing term, uniformly in $\e$, in the thermodynamic limit, being $\lim_N \frac{1}{N}\log\left(\frac{S_N}{\sqrt{2\pi}^N}\right)=\log R+\frac{1}{2}$.

\vspace{0.3cm}

\subsection{Equivalence of Spherical and Gaussian Ensemble}


\ni In the last step of our analysis we exploit the equivalence of ensembles of statistical mechanics (a classical reference is \cite{Ru69}), here seen as spherical and Gaussian ensembles. Hereafter for simplicity we will set once for all $\bar \l=1$. 

\ni Let us introduce the annealed pressure for the Gaussian model
\be
a^g(\b,\l):=\lim_{N\to\infty}\frac 1 N \log \E Z^{g}_N.
\ee
Then we have the following
\begin{proposition}
We have the following bounds:
\be\label{eq:AnnealingG}
\limsup_N A^g_N(\b, \l)\leq a_N^g(\b,\l)\leq \max\left(\frac12\left(2+\frac{\l}{\b}\right)^2, -\frac12\log(1-\l) \right).
\ee
\end{proposition}
\begin{remark}
The annealed bound is finite for $\l<1$, as shown also in \cite{BGGT}.
\end{remark}
\begin{proof}
We proceed exactly as in the proof of Proposition \ref{lem-annealed-sf}: at first we use Jensen inequality to exchange the logarithm and the expectation with respect to the quenched disorder; then we compute
$$
\E[Z^g_N(\b,\l)]=\E[Z^g_N(\b,\l)\mathcal{I}_{\B_a}]+\E[Z^g_N(\b,\l)\mathcal{I}_{\B^c_a}].
$$
For the first addendum we profit from the boundedness condition of the spectrum. We have that the maximum of the function $$\exp{\left(\beta \|z\|^2 -\frac{\beta^2}{4N}\|z\|^4+\frac{\lambda}{2}\|z\|^2\right)}$$ is attained for $\|z\|^2=N\frac{2\beta+\l}{\b^2}$ and it is equal to $\exp\left[\frac N2 \left(2+\frac{\l}{\b}\right)^2\right]$. Thus we have

\be\label{eq:AAnnGauss-B1}
\E[Z^g_N(\b,\l)\mathcal{I}_{\B_a}]\leq \exp\left[\frac N2 \left(2+\frac{\l}{\b}\right)^2\right].
\ee

\ni For the second addendum we can repeat the argument of Proposition \ref{lem-annealed-sf} to obtain
\be\label{eq:AAnnGauss-B2}
\E[Z^g_N(\b,\l)\mathcal{I}_{\B^c_a}]\leq \int_{\R^N}\frac{d^Nz}{(2\pi)^{\frac N2}}e^{-\frac{1-\l}{2}\|z\|^2}=e^{-\frac{N}{2}\log(1-\l)}.
\ee
Recollecting the contribution given by (\ref{eq:AAnnGauss-B1}) and (\ref{eq:AAnnGauss-B2}), we have (\ref{eq:AnnealingG}) in the limit $N\to\infty$.
\end{proof}

\ni Then we want to show that the Gibbs measure of the Gaussian Model is concentrated on a ball of radius growing as $\sqrt{N}$. To this purpose, fixed two arbitrary numbers $\d>0$ and $R>0$, we set
\bea
T_N(\d,R)&:=&\left\{ z_1,...,z_N\in\R^N: \|z\|_N^2\geq R^2N^{1+\d} \right\},\label{eq:T_N}\\
Z^g_{\|z\|^2\geq R^2 N^{1+\d}}&:=&\int_{T_N(\d,R)}dz_1...dz_N\frac{e^{-\|z\|^2/2}}{(2\pi)^{N/2}}e^{\left(-\beta H_N(z,J)-\frac{\beta^2}{4N}\|z\|^4+\frac{\lambda}{2}\|z\|^2\right)}.
\eea

\ni So we are ready to establish the subsequent

\begin{lemma}\label{pr:tail}
Let us fix arbitrarily $R>0$. For every $\d>0$
\be\label{eq:concentrato}
Z^g_{\|z\|^2\geq R^2 N^{1+\d}}= \OOO{e^{-N^{1+\d}}}.
\ee
\end{lemma}

\begin{proof}
\ni To begin with, let us define for a certain event $\Omega\subseteq\R^N$
\be
\pi^{\b,\l}_N(\Omega):=\frac{1}{Z^g_N}\int_\Omega \g_N(z) \exp\left(\b\sum_i^N \l_iz_i^2+\l\|z\|^2-\frac{\b}{4N}\|z\|^4\right).
\ee  
Let $\mathbb{E}_\pi$ be the expectation w.r.t. the density $\pi^{\b,\l}_N$ and $R,\d>0$ fixed. Using the Chernoff bound we get for every $\mu>0$
\be
\pi^{\b,\l}_N(\|z\|^2\geq R^2N^{1+\d})\leq e^{-\frac{\mu}{2} R^2 N^{1+\d}}\mathbb{E}_\pi\left[e^{\mu\|z\|^2}\right].
\ee
In particular we choose $\max(0,\l-1)<\mu\leq\l$. In addition we have
\bea
\log\mathbb{E}_\pi\left[e^{\mu\|z\|^2}\right]&=&\log\int_{T_N(\d,R)}dz_1...dz_N\frac{e^{-\|z\|^2/2}}{(2\pi)^{N/2}}e^{\left(-\beta H_N(z,J)-\frac{\beta^2}{4N}\|z\|^4+\frac{\lambda}{2}\|z\|^2\right)}-N\log Z_N^g(\b,\l)\nn\\
&=&N[A_N(\b, \l-\mu)-A_N(\b,\l)]\nn.
\eea
Therefore
\bea
Z^g_{\|z\|^2\geq R^2 N^{1+\d}}&\leq&\exp\left[ -\mu R^2N^{1+\d}+NA_N(\b,\l-\mu) \right]\nn\\
&=&\exp\left[ -N^{1+\d}\left(\mu R^2-N^{-\d}A_N(\b,\l-\mu)\right) \right]\nn\\
&\leq&\exp\left[ -N^{1+\d}\max_{\mu\in(0,\l]}\left(\mu R^2-N^{-\d}A_N(\b,\l-\mu)\right) \right],\nn
\eea
and, because of the annealed bound (\ref{eq:AnnealingG}), we obtain (\ref{eq:concentrato}).
\end{proof}

\ni Finally we can proceed with the

\begin{proof}[Proof of (\ref{eq:A-gauss})]
We will show that the r.h.s. of (\ref{eq:A-gauss}) is an upper and a lower bound for $A^g(\b,\l)$.

\ni We start with the lower bound. Let us define $\bar{R}$ as the radius where it is reached the unique maximum of $(\ref{eq:A-gauss})$ and let us abbreviate hereafter $S_R^\e:=S_{R\sqrt{N}}^\e$.
We have that 
\begin{eqnarray}
Z^g_N(\b,\l;J) &=&\int_{\R^N}dz_1...dz_N\frac{e^{-\|z\|^2/2}}{(2\pi)^{N/2}}e^{\left(-\beta H_N(z,J)-\frac{\beta^2}{4N}\|z\|^4+\frac{\lambda}{2}\|z\|^2\right)}\nonumber\\
&\geq & \int_{S_{\bar{R}}^{\e}}dz_1...dz_N\frac{e^{-\|z\|^2/2}}{(2\pi)^{N/2}}e^{\left(-\beta H_N(z,J)-\frac{\beta^2}{4N}\|z\|^4+\frac{\lambda}{2}\|z\|^2\right)}\nonumber\\
&\gtrsim& e^{N\left(\frac{(\l-1)}{2}\bar{R}^2-\frac{\beta^2}{4}\bar{R}^4+\frac 1 N \log |S_{\bar{R}\sqrt{N}}|-\frac 1 2 \log (2\pi)\right)}\left(\int_{S_{\bar{R}}^{\e}}\frac{dz_1...dz_N}{\e |S_{\bar{R}\sqrt{N}}|}e^{\left(-\beta H_N(z,J)\right)}\right)\nn
\end{eqnarray}
and then $\forall \e$,
\be
A^g_N(\b,\l)\geq \frac{(\l-1)}{2}\bar{R}^2-\frac{\beta^2}{4}\bar{R}^4+\frac 1 N \log |S_{\bar{R}\sqrt{N}}|-\frac 1 2 \log (2\pi) +A_{N,\e}^{sh}(\beta,\bar{R}\sqrt{N}).
\ee
We take the $\liminf$ over $N$ on the left, and just the limit for $N\to\infty$ on the right. By using Lemma $\ref{lem:chlim}$ and the computation $\frac 1 N \log |S_{\bar{R}\sqrt{N}}|-\frac 1 2 \log (2\pi) \to \log\bar{R} +\frac 1 2 $, we eventually obtain the r.h.s. of (\ref{eq:A-gauss}) as a lower bound:
\be\label{eq:liminfAg}
\liminf_N A^g_N(\b,\l)\geq \max_{R\in(0,\infty)}\left(A^{sf}(\b,R)-\frac{\b^2R^4}{4}+\frac{(\l-1)R^2}{2}+\log R+\frac{1}{2} \right).
\ee 

\ni In order to get the reverse bound, for an arbitrary $\delta>0$ we decompose $\R^N=T_N(\delta)\cup T^c_N(\delta)$. In virtue of Proposition $\ref{pr:tail}$ the integration over $T_N(\delta)$ does not give any thermodynamical contribution to the free energy. Consequently, for simplicity, we can consider as configuration space just $T^c_N(\delta)$. Then we look at a generic partition of $ T_N^c(\delta)$ into $N^{\delta}/2\e$ shells of thickness $2\e$ and we estimate
\begin{eqnarray}
Z^g_N &\leq& \frac{N^{\delta}}{2\e}\max_{R\in[0,N^{\delta}]}\left(\int_{S_{R}^{\e}}dz_1...dz_N\frac{e^{-\|z\|^2/2}}{(2\pi)^{N/2}}e^{\left(-\beta H_N(z,J)-\frac{\beta^2}{4N}\|z\|^4+\frac{\lambda}{2}\|z\|^2\right)}\right)\nonumber \\
&=&\frac{N^{\delta}}{2\e}\max_{R\in[0,N^{\delta}]}\left(e^{N\left(\frac{(\l-1)}{2}R^2-\frac{\beta^2}{4}R^4+\frac 1 N \log S^N_{R}-\frac 1 2 \log (2\pi)\right)+ o_N(1)}\int_{S_{R}^{\e}}\frac{dz_1...dz_N}{\e |S_{R\sqrt{N}}|}e^{\left(-\beta H_N(z,J)\right)}\right).\nn
\end{eqnarray}
Thus
\begin{eqnarray}
A^g_N&\leq& \frac{\delta}{N}\log\left(\frac{N}{2\e}\right)\nonumber \\
&+&\max_{R\in[0,N^{\delta}]}\left(\frac{(\l-1)}{2}R^2-\frac{\beta^2}{4}R^4+\frac 1 N \log |S_{R\sqrt{N}}| -\frac 1 2 \log (2\pi) +A_{N,\e}^{sh}(\beta,R\sqrt{N}) +o_N(1)\right)\nn.
\end{eqnarray}
Again we take $N\to\infty$: the $\limsup$ on the left and the limit on the right. Here we note that it is possible to exchange the limit with the $\max$, since the functional converges uniformly in $R$ in each bounded subset of $\R$ and tends to $-\infty$ as $R\to\infty$ for all $N$. Thus we get the reverse inequality:
\be\label{eq:limsupAg}
\limsup_N A^g_N(\b,\l)\leq \max_{R\in(0,\infty)}\left(A^{sf}(\b,R)-\frac{\b^2R^4}{4}+\frac{(\l-1)R^2}{2}+\log R+\frac{1}{2} \right).
\ee

\ni Now we notice that the pressure of the spherical model $A^{sf}(\b,R)$ depends in fact on $R^2$:
\bea
Z^{sf}_N(\b,R)&=&\int_{\|z\|^2=R^2N}\frac{dz_1...dz_N}{|S_{R\sqrt{N}}|} e^{-\beta \sum_i \l_i z_i^2}\nn\\
&=&\int_{\|\bar z\|^2=N}\frac{d\bar z_1...d\bar z_N}{|S_N|} e^{-\beta R^2 \sum_i \l_i \bar z_i^2}\nn\\
&=&Z^{sf}_N(\b R^2,1)\nn,
\eea
by the simple change of variables $R \bar z=z$. So we have
$$
A^{sf}(\b,R)=A^{sf}(\b R^2,1),
$$
and
\be\label{eq:main-formula-v2}
A^g(\b,\l)=\max_{R^2\in(0,\infty)}\left(A^{sf}(\b R^2,1)-\frac{\b^2R^4}{4}+\frac{(\l-1)R^2}{2}+\frac{\log R^2+1}{2} \right).
\ee
\end{proof}

\vspace{0.5cm}


\section{Conclusions and Outlooks}\label{sec:concl}

\ni In this paper we have analysed the relation between Gaussian and spherical spin glass models. In particular we have pointed out precisely their duality in terms of Legendre structure, or equivalence of Gaussian and spherical ensembles. Our work consequently permits to deal deliberately with one or the other model in further studies. 

\ni It is worthwhile to remark that the explicit representation ($\ref{gaussian free energy}$) coincides with the RS approximation exhibited in \cite{BGGT}. This enables us to complete the picture, by identifying $\bar q$ with the Edward-Anderson order parameter of the model. We have that replica symmetry holds in the whole phase diagram and the transition is between a high temperature phase and a RS one. It is ruled by the value of the overlap, that is fixed to zero in $\{(\b,\l):\: \b<1-\l,\, \l<1\}$ and to $\bar q$ otherwise. This holds true for random interactions in the Wigner ensemble with a sub-Gaussian tail, according to the hypothesis \textbf{H}.

\ni We stress that our approach, reminiscent of the earliest works in spin glasses, albeit apparently more general, does not give naturally this picture. Practically, we do not know a priori the significance of the minimiser $\bar q$. We need the scheme of \cite{BGGT} for giving a complete interpretation to our results in terms of the correct order parameter (\ie the overlap). Of course, Legendre duality permits us to transfer all these considerations to the spherical spin glass as well.

\ni Lastly, we examine the relation between the Gaussian and spherical models and the (analogical) Hopfield model of neural networks. The analogical Hopfield Model is defined as follows: consider $N$ Bernoulli spin r.vs $\s_i$ interacting via the Hamiltonian
\be
H_N=-\frac{1}{N}\sum_{\mu=1}^{K}\sum_{i,j}^N\xi_i^{\mu}\xi^{\mu}_j\s_i\s_j
\ee
where $\xi_i^{\mu}$ are $K$ $\mathcal{N}$(0,1) i.i.d. quenched random vectors in $\R^N$ (or patterns) with $\lim_N K/N=\alpha \in \R^+$. One is interested as usual to the pressure of the model defined as
$$
A_H(\b,\a)=\lim_{N,K} \frac1N\mathbb{E} \log \sum_\s e^{-\b H_N}.
$$

\ni In the original formulation in the celebrated paper by Hopfield \cite{hop}, the random patterns were Bernoulli $\pm1$ r.vs. The two versions are in fact supposed to be qualitatively different \cite{univ}, and we refer to \cite{bovbook} and \cite{talabook} for an exhaustive account on the topic.

\ni By a Gaussian transformation we can map the Hopfield model in a bipartite spin glass with Bernoulli and Gaussian spin:
$$
-\frac{1}{N}\sum_{\mu=1}^{K}\sum_{i,j}^N\xi_i^{\mu}\xi^{\mu}_j\s_i\s_j\longrightarrow-\frac{1}{\sqrt{N}}\sum_{\mu=1}^{K}\sum_{i}^N\xi_i^{\mu}\s_iz_\mu.
$$
\ni With this approach, it has been shown in \cite{NN?} that the pressure of the analogical Hopfield Model, at least in RS regime, can be written as a convex combination of the pressure of a SK model and the one of a Gaussian model, calculated at different suitable temperatures. More precisely, let us define for $\bar q_H>0$
\bea
\beta_1 &:=& \frac{\sqrt{\alpha} \beta}{1 - \beta(1- \bar{q}_H)}, \label{b1}\\
\beta_2 &:=& 1 - \beta(1-\bar{q}_{H})\label{b2}.
\eea
We have proven that, fixed $\b_1$ and $\b_2$ as in (\ref{b1}) and (\ref{b2}), the replica symmetric approximation of the quenched pressure of the analogical neural networks can be linearly decomposed as follows:
\be\label{eq:dec-SK-G}
A_{NN}^{RS}(\beta) = A_{SK}^{RS}(\beta_1) -\frac14 \beta_1^2 + \alpha A_{Gauss}(\beta_2,\b).
\ee

\ni Now we can rephrase this result in terms of the spherical model by means of the duality we have established. We note that the radius of the spherical model has a definite meaning in the context of neural networks, being the self-overlap of the Gaussian spins $p_{11}:=\frac 1 K \sum_{i=1}^K z_i^2$. This is known to be self averaging and related to the internal energy \cite{Salerno}: 
\be
\lim_{N\to\infty} \frac{\left\langle H_N\right\rangle}{N}=\frac{\alpha}{2\beta}\left(1-\left\langle p_{11}\right\rangle\right).
\ee

\ni We have

\begin{corollario}
Fixed $\b_1$ and $\b_2$ as in (\ref{b1}) and (\ref{b2}), the replica symmetric approximation of the quenched pressure of the analogical neural networks can be linearly decomposed as follows
\bea
A_{NN}^{RS}(\beta)&=&A_{SK}^{RS}(\beta_1) -\frac14 \beta_1^2 + \alpha A^{sf}(\beta_2,\sqrt{p_{11}})\nn\\
&+&\frac{\a}{2}(\b-1)\meanv{p_{11}}-\frac{\a\b_2^2}{4}\meanv{p^2_{11}}+\frac{\a}{2}(1+\log \meanv{p_{11}})\label{eq:decSK-Sf}.
\eea
\end{corollario}

\ni Here we stress that, even though the r.h.s. of (\ref{eq:dec-SK-G}), (\ref{eq:decSK-Sf}) are in fact unaffected by the choice of the randomness according to the hypothesis \textbf{H}, the decomposition by itself has been proven only in the case of Gaussian disorder. We do not know in particular whether it can be extended to the original Hopfield model with $\pm 1$ patterns.

\ni In conclusion, the invariance properties under rotations of the Hopfield model, and so its connection with rotationally invariant spin glasses, are indubitably suggestive. They have been certainly already investigated, but probably they still can be useful tools in order to shed more light on its mathematical structure.

\vspace{0.5cm}

{\bf \ni Acknowledgements\\}
It is a pleasure to thank Francesco Guerra, source of inspiration of many ideas in this work. We are also grateful to G\'erard Ben Arous for an enlightening discussion on the paper \cite{BDG}, to Benjamin Schlein for some precious advices 
and to GNFM (Gruppo Nazionale per la Fisica Matematica). Finally we thank two anonymous referees for the suggestions that led to significant improvements of the paper. G.G. is supported by the ERC Grant MAQD 240518. D.T is partially supported by "Avvio alla Ricerca 2014", Sapienza University of Rome.


\begin{thebibliography}{9}

\bibitem{AGZ} G. Anderson, A. Guionnet, O. Zeitouni, {\em An Introduction to Random Matrices}, Cambridge University Press, (2010).

\bibitem{Salerno} A. Barra, F. Guerra, {\em Locking of order parameters in analogical associative neural networks}, Percorsi Incrociati, Antonio Vitolo et al., eds. (2009).

\bibitem{Bip} A. Barra, G. Genovese, F. Guerra, \textit{Equilibrium statistical mechanics of bipartite spin systems}, J. Phys. A: Math. Theor. \textbf{44}, 245002 (2011).

\bibitem{BGGT} A. Barra, G. Genovese, F. Guerra, D. Tantari, {\em About a solvable mean field model of a Gaussian spin glass}, J. Phys. A: Math. Theor. \textbf{47}, 155002, (2014);

\bibitem{NN?} A. Barra, G. Genovese, F. Guerra, D. Tantari, \textit{How glassy are neural networks?}, J. Stat. Mech. P07009 (2012).

\bibitem{BDG} G. Ben Arous, A. Dembo, A. Guionnet, {\em Aging of spherical spin glasses}, Prob. Theor. Related Fields 120, 1, (2001).

\bibitem{BK} T. H. Berlin, M. Kac, {\em The Spherical Model of a Ferromagnet}, Phys. Rev. 86, 821 (1952).

\bibitem{bovbook} A. Bovier, {\em Statistical mechanics of disordered system. A mathematical perspective}, Cambridge University Press, (2006).

\bibitem{CS} A. Crisanti, H. -J. Sommers, {\em The spherical p-spin interaction p-spin glass model: the statistics}, Z. Phys. B Condensed Matter 83, 341, (1992).

\bibitem{persi} P. Diaconis, D. Freedman. {\em A dozen de Finetti-style results in search of a theory}, Ann. Inst. H. Poincar\'e Prob. et Stat., \textbf{23}, 397, (1987). 

\bibitem{FS} Y. V. Fyodorov, H. -J- Sommers, {\em Classical Particle in a Box with Random Potential: exploiting rotational symmetry of replicated Hamiltonian}, Nuclear Physics B 764,128 (2007).

\bibitem{franz} S. Franz, F. Tria, {\em A Note on the Guerra and Talagrand Theorems for Mean Field Spin Glasses: The Simple Case of Spherical Models}, J. Stat. Phys. \textbf{122}, 313, (2005).

\bibitem{univ} G. Genovese, {\em Universality in Bipartite Mean Field Spin Glasses}, J. Math. Phys. \textbf{53}, 123304, (2012);

\bibitem{leshouches} F. Guerra, {\em An Introduction to Mean Field Spin Glass Theory: Methods and Results}, A. Bovier et al. eds, Les Houches, Session LXXXIII, (2005).

\bibitem{hop} J.J. Hopfield, {\em Neural networks and physical systems with emergent collective computational abilities}, Proc. Nat. Acad. Sci. USA \textbf{79},  2554-2558 (1982).

\bibitem{KT} J. M. Kosterlitz, D. J. Thouless, Raymund C. Jones, \textit{Spherical model of a spin glass}, Phys. Rev. Lett. \textbf{36}, 1217 (1976).

\bibitem{Mk} H. P. McKean {\em Geometry of Differential Spaces}, Ann. Prob. \textbf{1}, 197, (1973).

\bibitem{Mon} E. W. Montroll {\em Continuum Models of Cooperative Phenomenon}, Il Nuovo Cimento \textbf{VI}, 264, (1949).

\bibitem{Panc} D. Panchenko, {\em Cavity method in the spherical SK model}, Ann. Inst. H. Poincar\`e, Prob. et Stat. \textbf{45}, 1020, (2009).

\bibitem{PT} D. Panchenko, M. Talagrand, {\em On the Overlap in the Multiple Spherical Models}, Ann. Probab. \textbf{35}, 2321 (2007).

\bibitem{Ru69} D. Ruelle, {\em Statistical Mechanics. Rigorous results}, W.A. Benjamin Inc., New York, 1969.

\bibitem{Tsfer} M. Talagrand, \textit{Free energy of the spherical mean field model}, Prob. Theor. Related Fields \textbf{134}, 339 (2006).

\bibitem{talabook} M. Talagrand, \emph{Mean Field Models for Spin Glasses}, Vol. 1,2, Springer-Verlag Berlin Heidelberg (2011).

\bibitem{Tao} T. Tao, {\em Topics in Random Matrix Theory}, AMS Graduate Studies in Mathematics, Providence, Rhode Island (2012).

\bibitem{VN} J. Von Neumann, {\em Distribution of the Ratio of the Mean Square Successive Difference to the Variance},
Ann. Math, Stat. \textbf{12}, 367, (1941).

\end{thebibliography}
\end{document}